\documentclass[preprint,a4paper,12pt]{elsarticle}

\usepackage[utf8]{inputenc}
\usepackage[english]{babel}
\usepackage{slantsc} 
\usepackage[mathscr]{euscript}
\usepackage{relsize}
\usepackage{amssymb}
\usepackage{hyperref}
\usepackage{xcolor}
\usepackage{amsmath}
\usepackage{amsfonts}
\usepackage{amssymb}
\usepackage{amsthm}
\usepackage{graphicx}
\usepackage{hyperref}
\usepackage{xcolor}

\definecolor{gold}{HTML}{FDD017}
\definecolor{orange}{HTML}{FF8000}
\definecolor{deep sky blue}{HTML}{3BB9FF}
\definecolor{blue}{HTML}{0101DF}
\definecolor{red}{HTML}{DF0101}
\definecolor{co}{HTML}{2020A0}

\swapnumbers

\theoremstyle{theorem}
\newtheorem{thm}{Theorem}[section]
\newtheorem{prop}[thm]{Proposition}
\newtheorem{cor}[thm]{Corollary}

\theoremstyle{definition}

\newtheorem{rem}[thm]{Remark}

\newcommand{\seq}[3]{#1_{#2},\dotsc,#1_#3}
\newcommand{\floor}[1]{\lfloor#1\rfloor}

\newcommand{\FF}{\mathbb{F}}

\newcommand{\bs}[1]{\boldsymbol{#1}}
\newcommand{\diag}{\textrm{\upshape diag}}
\newcommand{\RS}{\textrm{RS}}
\newcommand{\GRS}{\textrm{GRS}}
\newcommand{\BCH}{\textrm{BCH}}
\newcommand{\wt}{\widetilde}

\begin{document}

\begin{frontmatter}

\title{On PGZ decoding of alternant codes}

\author{R. Farré, N. Sayols, and S. Xambó-Descamps}
\address{Universitat Politècnica de Catalunya\\
\textsf{\upshape\footnotesize rafel.farre@upc.edu},
\textsf{\upshape\footnotesize narcissb@gmail.com},
\textsf{\upshape\footnotesize sebastia.xambo@upc.edu}
}

\begin{abstract}
In this note we first review the classical
Petterson-Gorenstein-Zierler decoding algorithm for the class of
alternant codes (which includes Reed-Solomon, Bose-Chaudhuri-Hocquenghem
and classical Goppa codes), then we present an improvement
of the method to find the number of errors and
the error-locator polynomial, and finally we 
illustrate the procedure with several examples. 
In two appendices we sketch the main features of the 
system \cite{Sayols-Xambo-2017}
we have designed and developed for the computations.
\end{abstract}

\begin{keyword}
Alternant codes\sep 
RS codes\sep 
BCH codes\sep classical Goppa codes 

\MSC[2010] 
11T71\sep 
94B05\sep 
94B35\sep 
94B15
\end{keyword}

\end{frontmatter}

\section*{Introduction}
The Petterson-Gorenstein-Zierler decoding algorithm (PGZ for short)
was first developed for Reed-Solomon codes (RS), and later
applied to Bose-Chaudhuri-Hocquenghem codes (BCH).
In \cite{Xambo-2003}, two flavours of it were presented for alternant codes,
with due attention to the computational aspects. 
The main interest of working with the class of alternant codes is that 
it includes many interesting subclasses, like 
RS codes, BCH codes (the most relevant class of cyclic codes), 
and classical Goppa codes. 
The practical bonus of this
realization is that all these families of codes can be
constructed by specializing the general constructor of
alternant codes and, most fundamentally, that any
effective decoding algorithm for alternant codes is 
sufficient (and effective) for all those subclasses. 

In this note we present a natural improvement, both
conceptual and computational, of the PGZ algorithm.
The key point is that the output of the Gauss-Jordan 
reduction of a (Hankel-like) matrix constructed from the syndrome vector
gives directly and at the same time the number of errors and
the error-locator polynomial.

The organization is as follows. In the first section we
briefly review alternant codes. This includes details about
how the classes of codes just mentioned can be contructed
with special calls to the main constructor. 
The second section is devoted
to present the mathematical basis of the PGZ approach for
the decoding of alternant codes. Our improvement of PGZ is explained in
detail in the third section and in the fourth we provide
several examples. Finally
Appendix A contains listings of the key
Python functions that we have designed and coded to get 
clear implementations of the computations
and in the Appendix B we sketch the main
features of the Python package 
\cite{Sayols-Xambo-2017}
used to script the examples.

\paragraph{Notations and conventions}

If $q$ is a prime power, the finite field of $q$ elements (unique up to
isomorphism) is denoted $\FF_q$. It is a subfield of $\FF_{q^m}$ for all
positive integers $m$. The field $\FF_{q^m}$ can be constructed as the
quotient $\FF_q[X]/(f)$, where $f\in \FF_q[X]$ is any irreducible polynomial 
over $\FF_q$ of degree~$m$.

Given elements $\seq{\alpha}1n$ in a ring, 
we write $V_r(\seq{\alpha}1n)$ to denote the 
\textit{Vandermonde matrix} of $r$ rows associated to $\seq{\alpha}1n$.
In other words, its rows have the form $(\seq{\alpha^i}1n)$
for $0\le i\le r-1$. The determinant of the matrix
$V_n(\seq{\alpha}1n)$, which is called the 
\textit{Vandermonde determinant} of $\seq{\alpha}1n$, is equal to 
$\prod_{1\le i<j\le n}(\alpha_j-\alpha_i)$. In particular, 
it is non-zero when the $\alpha_k$ are distinct elements of a field.

\newcommand{\hd}{\textrm{hd}}\newcommand{\wgt}{\textrm{wt}}
Let $K$ be a finite field. A linear code of length $n$ defined over $K$ is a 
vector subspace $C\subseteq K^n$. If $C$ has dimension $k$, we say that
$C$ is an $[n,k]$ code. The quotient $k/n$ is called the \textit{rate} of $C$. 
The \textit{Hamming distance} 
$\hd(y,y')$ of $y,y'\in K^n$ is the number of indices $j\in \{1,\dotsc,n\}$ 
such that $y_j\ne y'_j$. The \textit{minimum distance} of a linear code
$C$, denoted $d$, is the minimum of the distances $\hd(x,x')$ for $x,x'\in C$, $x\ne x'$.
The number of non-zero entries of $y\in K$ is called the \textit{weight} of
$y$ and is denoted $\wgt(y)$. It is easy to see that $d$ 
is the minimum of the weights of 
non-zero elements of $C$. An $[n,k]$ code of minimum distance $d$ is said to be an $[n,k,d]$ code, or an $[n,k,d]_K$ if we need to write the base field $K$ explicitly.

\setcounter{section}{0}

\section{Essentials on alternant codes}\label{sec:AC}

Let $K=\FF_q$ and $\bar{K}=\FF_{q^m}$. Let
$\seq{\alpha}1n$ and $\seq{h}1n$ be elements of $\bar{K}$
such that $h_i,\alpha_i \neq 0$ for all $i$ and $\alpha_i\neq\alpha_j$ for all $i\neq j$.
Consider the matrix
\begin{equation}
\label{df-alt-matrix-short}
  H=V_r(\seq{\alpha}1n)\diag(\seq{h}1n)\in M_n^r(\bar{K}),
\end{equation}
that is,
\begin{equation}
\label{df-alt-matrix}
   H=\begin{pmatrix} h_1 &
   \dots & h_n\\ h_1\alpha_1 & \dots & h_n\alpha_n\\ \vdots & &
   \vdots \\ h_1\alpha_1^{r-1} & \dots & h_n\alpha_n^{r-1}
  \end{pmatrix}
\end{equation}
We say that $H$ is the
\emph{alternant control matrix}
of order $r$ associated with the vectors
\[
\bs{h}=(\seq h1n) \quad \textrm{and} \quad
\bs{\alpha}=(\seq{\alpha}1n).
\]
To make explicit that the entries of $\bs{h}$ and $\bs{\alpha}$
(and hence of $H$) lie in $\bar{K}$, we will say that
$H$ is defined over $\bar{K}$.

The $K$-code $A_K(\bs{h},\bs{\alpha},r)$
defined by the control matrix $H$ is the subspace
of $K^n$ whose elements are the vectors
$x$ such that $xH^T=0$.
Such codes will be called \emph{alternant codes}.
\label{alternantcodes}%
If we define the $H$\emph{-syndrome} of a vector
$y\in\bar{K}^n$ as $s=yH^{T}\in \bar{K}^r$, then
$A_K(\bs{h},\bs{\alpha},r)$ is just the
subspace of $K^n$ whose elements are the
vectors with zero $H$-syndrome.

\begin{prop}[Alternant bounds]\label{alternant-bound}
If $C=A_K(\bs{h},\bs{\alpha},r)$, then
\[n-r \ge \dim C \ge n-rm\]
and
\[d \ge r+1
\]
\emph{(minimum distance alternant bound)}.
\end{prop}
\begin{proof}
See, for example, \cite{Xambo-2003}, p. 183.
\end{proof}

For the proofs of the statements in the remainder of this section,
we refer to \cite{Xambo-2003}, Section 4.1.

\subsubsection*{Reed-Solomon codes}

Given a list or vector $\bs{\alpha}$ of
distinct non-zero elements $\seq{\alpha}1n\in K$, the Reed--Solomon code
\[C=\RS(\bs{\alpha}, k)\subseteq K^n
\] 
is the subspace of $K^n$
generated by the rows of the Vandermonde matrix 
$V_{k}(\seq{\alpha}1n)$. It turns out that
\[\label{rs<alternant}%
\RS(\bs{\alpha},k)=
A_K(\bs{h},\bs{\alpha},n-k),
\]
where $\bs{h}=(\seq h1n)$ is given by
\begin{equation}
\label{fm-h-rs}
   h_i=1/\prod_{j\neq i}(\alpha_j-\alpha_i).
\end{equation}
Note that in this case $\bar{K}=K$, hence $m=1$,
and that the alternant bounds are sharp. Indeed,
we have $r=n-k$, hence $k=n-r$, while
$n-k+1\ge d$ (by the Singleton bound)
and $d\ge r+1=n-k+1$ by the minimum distance alternant bound. 
In other words, $C$ is MDS (maximum distance separable).

An RS code is called \textit{primitive} if
the $\seq{\alpha}1n$ are all non-zero elements of $K$.
In that case, a natural way to proceed is to generate
those elements as the powers $1,\alpha,\dotsc,\alpha^{q-2}$
of a primitive element $\alpha$ of $K$ and so the code is,
if its dimension is $k$, $\RS([1,\alpha,\dotsc,\alpha^{n-1}],k)$,
where $n=q-1$.

\paragraph{Generalized Reed-Solomon codes}
The vector $\bs{h}$ in the definition of the code
$\RS([\seq{\alpha}1n],k)$ as an alternant code is obtained
from $\bs{\alpha}$ by the formula (\ref{fm-h-rs}).
If we allow that $\bs{h}$ can be chosen possibly unrelated
to $\bs{\alpha}$, but still with components in $K$,
the resulting codes $A_K(\bs{h},\bs{\alpha},n-k)$
are called \emph{Generalized Reed--Solomon} (GRS) 
codes, and we will write $\GRS(\bs{h},\bs{\alpha},k)$ to denote
them. An argument as above shows that 
such codes have type $[n,k,n-k+~1]$.
Notice that the code $A_K(\bs{h},\bs{\alpha},r)$
is the intersection of the GRS code 
$A_{\bar{K}}(\bs{h},\bs{\alpha},r)$
with $K^n$.

\subsubsection*{BCH codes}

These codes are denoted $\BCH(\alpha,d,l)$, where
$\alpha \in \bar{K}$ and $d>0$, $l\ge 0$ are integers
(called the \textit{design minimum distance} and
the \textit{offset}, respectively). When $l=1$, 
we simply write $\BCH(\alpha,d)$ and say that the 
it is a \textit{strict} BCH code.
The good news here is that
\begin{equation}\label{BCH}
\BCH(\alpha,d,l) = A_K(\bs{h},\bs{\alpha},d-1),
\end{equation}
where $\bs{h}=(1,\alpha^l,\alpha^{2l},\dotsc,\alpha^{(n-1)l})$,
$\bs{\alpha}=(1,\alpha,\alpha^{2},\dotsc,\alpha^{(n-1)})$,
$n=\textrm{period}(\alpha)$.

If $\alpha$ is a primitive element of $K$, and hence $n=q-1$,
we have the equality
\[
\BCH(\alpha, n-k+1)=\RS([1,\alpha,\dotsc,\alpha^{n-1}], k).
\]

\subsubsection*{Classical Goppa codes}

Let $g\in \bar{K}[T]$ be a polynomial of degree
$r>0$ and let $\bs{\alpha}=\seq {\alpha}1n\in \bar{K}$ be
distinct non-zero elements such that $g(\alpha_i)\neq0$ for all $i$.
Then the \emph{classical Goppa code} over $K$
associated with $g$ and $\bs{\alpha}$,
which will be denoted $\Gamma(g,\bs{\alpha})$,
\label{classicalgoppacodes}%
can be defined as $A_K(\bs{h},\bs{\alpha},r)$, where
$\bs{h}$ is the vector
$(1/g(\alpha_1),\dotsc,1/g(\alpha_n))$.
Thus the minimum distance of $\Gamma(g,\bs{\alpha})$
is $\ge r+1$ and its dimension $k$
satisfies $n-rm\le k\le n-r$.
The minimum distance bound can be improved to $d\ge 2r+1$
in the case that $K=\FF_2$ and the roots of $g$ 
are distinct.

\section{The PGZ decoding approach}

Let $C=A_K(\bs{h},\bs{\alpha},r)$ be an alternant code.
Let $t=\floor{r/2}$, that is, the highest integer
such that $2t\le r$. For reasons that will become
apparent below, $t$ is called the
\textit{error-correction capacity} of $C$.

Let $x\in C$ (using a transmission chanel terminology, 
we say that it is the \emph{sent vector}) and $e\in\bar{K}^n$
(\emph{error vector}, or \emph{error pattern}).
Let $y=x+e$ (\emph{received vector}).
The goal of a decoder is to obtain
$x$ from $y$ and $H$ when $l:=\wgt(e)\le t$.
Henceforth we will assume that $l>0$.

If $e_m\ne 0$, we say that $m$ is an \textit{error position}.
Let $\{m_1,\dotsc,m_l\}$ be the error positions
and $\{e_{m_1},\dotsc, e_{m_l}\}$ the corresponding \textit{error values}.
The \textit{error locators} $\eta_1,\dotsc,\eta_l$ are defined by
$\eta_k = \alpha_{m_k}$. Since $\alpha_1,\dotsc,\alpha_n$ are distinct, the
knowledge of the $\eta_k$ is equivalent to the knowledge of the error positions.

The monic polynonial $L(z)$ whose roots are the error locators
is called the \textit{error-locator polynomial}. Notice that
\begin{equation}
\label{eq:error-locator-pgz}
 L(z)=\prod_{i=1}^{l}(z-\eta_i)=z^l+a_1 z^{l-1} + a_2
 z^{l-2} +\cdots+a_l,
\end{equation}
where $a_j=(-1)^j\sigma_j=\sigma_j(\eta_1,...,\eta_l)$
is the $j$-th elementary symmetric polynomial in the $\eta_i$
$(0\leq j\leq l)$.

The \textit{syndrome} of $y$ is the vector $s=yH^T$, 
say $s=(s_0,\dotsc,s_{r-1})$. Since
$xH^T=0$, we have $s=eH^T$. Inserting the definitions, we easily find that
\begin{equation}
\label{eq:syndrome-formula}
s_j=\sum_{i=0}^{n-1} e_i h_i \alpha_i^j =
\sum_{k=1}^{l} h_{m_k}e_{m_k}\alpha_{m_k}^j=
\sum_{k=1}^{l}h_{m_k}e_{m_k}\eta_{k}^j
\end{equation}

We will use the following notations:
\begin{equation}
\label{A-ell}
A_{l}=
\begin{pmatrix}
s_0 & s_1 & \ldots & s_{{l}-1} \cr s_1 & s_2 & \ldots & s_{{l}}   \cr
\vdots & \vdots & \ddots & \vdots \cr 
s_{{l}-1} & s_{{l}} & \ldots & s_{2{l}-2} \cr
\end{pmatrix}
\end{equation}
and the vector 
\begin{equation}
\label{vector-b}
\bs{b}_l=(s_l,\dotsc,s_{2l-1}).
\end{equation}

Next proposition establishes the key relation for computing 
the error-locator polynomial. 

\begin{prop}
\label{As*a=-b}
If $\bs{a}_l=(a_l,...,a_1)$ \emph{(see the fomula (\ref{eq:error-locator-pgz}))},
then 
\begin{equation}\label{eq:coeffs-L}
\bs{a}_lA_l+\bs{b}_l=0.
\end{equation}
\end{prop}

\begin{proof}
Substituting $z$ by $\eta_i$ in the identity
\[
\prod_{i=1}^l(z-\eta_i)=z^l+a_1 z^{l-1}+...+a_l
\]
we obtain the relations
\[
\eta_i^l+a_1\eta_i^{l-1}+\cdots+a_l = 0,
\]
where $i=1,...,l$. Multiplying by $h_{m_i}e_{m_i}\eta_i^{j}$
and adding with respect to $i$, we obtain (using (\ref{eq:syndrome-formula})) the
relations
\[
s_{j+l}+a_1 s_{j+l-1}+\cdots+a_l s_j = 0,
\]
where $j=0,...,l-1$, and these relations are equivalent to
the stated matrix relation.
\end{proof}

\begin{rem}
In the equation (\ref{eq:coeffs-L}), the matrix $A_l$
turns out to be non-singular and hence it determines $\bs{a}_l$ (and $L(z)$) 
uniquely, namely $\bs{a}_l=-\bs{b}_lA_l^{-1}$. 
In next section we are going to establish this fact as a 
corollary of our Theorem \ref{thm:main}, whose main outcome is 
\textit{a fast solution of that equation}.
\end{rem}

The roots of $L$ only tell us in what positions the errors occur. To find the actual
value of the errors, we need the syndrome polynomial,
$\sigma(z)=s_0+s_1z+\cdots+s_{r-1}z^{r-1}$ and the polynomial
$\wt{L}(z)=1+a_1z+\cdots+a_lz^l$. Notice that the roots of $\wt{L}(z)$
are $1/\eta_1,\dotsc,1/\eta_l$.

\begin{thm}[Forney's formula]
Let $E(z)=\wt{L}(z) \sigma(z) \mod z^r$. Then for any
$m\in \{m_1,\dotsc,m_l\}$ we have
\begin{equation}\label{eq:Forney}
e_m = -\frac{\alpha_m\, E(1/\alpha_m)}{h_m {\wt{L}}^{'}(1/\alpha_m)},
\end{equation}
where ${\wt{L}}^{'}(z)$ denotes the derivative of $\wt{L}(z)$.
\end{thm}
\begin{proof}
See, for example, \cite{Xambo-2003}, sections 4.2 and 4.4.
\end{proof}
Because of this result, the polynomial $E(z)$ is called the
\textit{error-evaluator} polynomial.

\section{Fast PGZ computations}

The main object considered in this Section is the matrix
(cf. \cite{Farre-2003})
\begin{equation}
\label{eq:matrixS}
S=\begin{pmatrix}
s_0 & s_1 & \cdots & s_{l-1} & s_l & \cdots & s_t\\
s_1 & s_2 & \cdots & s_{l} & s_{l+1} & \cdots & s_{t+1}\\
\vdots & \vdots &  & \vdots & \vdots &  & \vdots\\
s_{l-1} & s_l & \cdots & s_{2l-2} & s_{2l-1} & \cdots & s_{t+l-1}\\
\hline
\vdots & \vdots &  & \vdots & \vdots &  & \vdots\\
s_{t-1} & s_t & \cdots & s_{t+l-2} & s_{t+l-1} & \cdots & s_{2t-1}
\end{pmatrix}
\end{equation}
Note that $2t-1\le r-1$, so that all components are well defined.
Note also that the $l\times l$ submatrix at the upper left corner is
the matrix $A_l$ defined by Eq. (\ref{A-ell}) an that the column 
$(s_l,s_{l+1},\dotsc,s_{2l-1})^T$ to its right is the vector $\bs{b}_l$
defined by Eq. (\ref{vector-b}). 

In next Theorem we use the following notation:
$V_s=V_s(\eta_1,\dotsc,\eta_l)$. Thus the $i$-th row of $V_s$, for
$0 \le i\le s-1$, is the vector $(\eta_1^i,\dotsc,\eta_l^i)$.
We also write $D=\diag(h_{m_1}e_{m_1},\dotsc,h_{m_l}e_{m_l})$.

\begin{thm}
\label{thm:main}
$S=V_tDV_{t+1}^T$.
\end{thm}

\begin{proof}
Let $0\le i \le t-1$ and $0\le j \le t$.
Then the $j$-th column of $DV_{t+1}^T$ is
the column vector $(h_{m_1}e_{m_1}\eta_1^j,\dotsc,h_{m_l}e_{m_l}\eta_l^j)^T$.
It follows that the element in row $i$ column $j$ of $V_tDV_{t+1}^T$
is $h_{m_1}e_{m_1}\eta_1^{i+j}+\cdots+h_{m_l}e_{m_l}\eta_l^{i+j}=s_{i+j}$
(by Eq.~(\ref{eq:syndrome-formula})).
\end{proof}

\begin{cor}
The rank of $S$ is $l$ and the matrix $A_l$ is non-singular.
\end{cor}

\begin{proof}
Since $D$ has rank $l$, the rank of $S$ is at most $l$.
On the other hand, the theorem shows that
$A_l = V_lDV_l^T$ and therefore
\[
\det(A_l) = \det(V_l)^2\det(D)\ne 0. 
\]
Note that $\det(V_l)$ is the Vandermonde determinant
of $\eta_1,\dotsc,\eta_l$, which is non-zero because
the error locators are distinct.
\end{proof}

\begin{cor}
The Gauss-Jordan algorithm applied to the matrix $S$ returns a matrix that has the form
\begin{equation}
\label{eq:matrixS2}
\begin{pmatrix}
1 & 0 & \cdots & 0 & -a_l &  * \\
0 & 1 & \cdots & 0 & -a_{l-1} &  *\\
\vdots & \vdots &  & \vdots & \vdots  & \vdots\\
0 & 0 & \cdots & 1 & -a_1 &  *\\
\hline
\vdots & \vdots &  & \vdots & \vdots & \vdots
\end{pmatrix}
\end{equation}
where $*$ denotes unneeded values \emph{(if any)} and
the vertical dots below the horizontal line denote that all its elements
\emph{(if any)} are zero. This matrix gives at the same time $l$, the number of errors, 
and the coefficients of the error-locator polynomial.\qed
\end{cor}

Putting together what we have learned in the last two sections, 
we obtain two algorithms to decode alternant codes, or rather two
variants of an algorithm. We call them \textsf{PGZ} and \textsf{PGZm},
for in essence they are due to Peterson, Gorenstein and
Zierler (see \cite{Peterson-Weldon-1972}). They share the same scheme
for finding the location of the errors, but differ in how
the error values are computed. PGZm is the simplest of the two,
as it relies mainly on linear algebra, whereas
PGZ relies on the finding the error evaluator polynomial
and using Forney's formula. 

In the descriptions that follow, \textit{Error} means ``a suitable decoding-error message''
and the function \textsf{GJ(S)}
returns the values $-a_l,\dotsc,-a_1$ of the matrix
(\ref{eq:matrixS2}) as a column vector
(this is a slightly modified form of the Gauss-Jordan procedure).
In detail, it works as follows:

\renewcommand{\ni}{\noindent}

\medskip
\ni
\textit{Improved PGZ}
\begin{enumerate}
\item
Get the syndrome vector, $s=(s_0,...,s_{r-1})=yH^{T}$.
If $s=0$, return $y$.
\item
Form the matrix $S$ as in the Eq. (\ref{eq:matrixS}).
\item 
Set $\boldsymbol{a}=-\textsf{GJ}(S)$ (Eq. (\ref{eq:matrixS2})).
After this we have $a_1,...,a_l$, hence also
the error-locator polynomial $L$.
\item 
Find the elements $\alpha_j$ that are roots of the polynomial
$L$. If the number of these roots is $<l$,
return \textit{Error}. Otherwise let
$\eta_1,...,\eta_l$ be the error-locators corresponding to the roots
and set $M=\{m_1,\dotsc,m_l\}$, where $\eta_i=\alpha_{m_i}$.
\item 
Let $\sigma(z)=s_0+s_1z+\cdots+s_{r-1}z^{r-1}$,
$\wt{L}(z)=1+a_1z+\cdots+a_sz^s$ and compute the error-evaluator 
polynomial by the formula
\[
E(z)=\wt{L}(z)\sigma(z)\mod z^r.
\]
\item 
Find the errors $e_{m}$, for all $m\in M$, using Forney's formula (equation
(\ref{eq:Forney})). If any of the values of $e_m$ is not in $K$, return \textit{Error}.
Otherwise return $y-e$.
\end{enumerate}

\begin{thm}
The algorithm PGZ corrects up to $t$ errors
\end{thm}
\begin{proof}
It is an immediate consequence of what we have seen so far.
\end{proof}
\begin{rem}
In step 5 of the algorithm we could use the alternative syndrome polynomial
$\wt{\sigma}(z)=s_0z^{r-1}+s_1z^{r-2}+\cdots+s_{r-1}$, find the alternative error
evaluator $E^*$ as the remainder of the division of
$L(z)\wt{\sigma}(z)$ by $z^r$ and then, in step 6, use the following
alternative Forney formula (\cite{Xambo-2003}, P.4.9):
\begin{equation}\label{eq:Forney-alt}
e_m = -E^*(\alpha_m)/h_m \alpha_m^r L'(\alpha_m).
\end{equation}
\end{rem}

\ni \textit{Algorithm PGZm}

\smallskip
The steps 5 and 6 of the PGZ algorithm can be compressed into 
a single step consisting in solving for
$e_{m_1},...,e_{m_l}$ the following system of linear equations:
\[
h_{m_1}e_{m_1}\eta_1^j+h_{m_2}e_{m_2}\eta_2^j+...+h_{m_l}e_{m_l}\eta_l^j=s_j
\ (0\le j\le l-1),
\]
which is equivalent to the matrix equation
\[
\begin{pmatrix} h_{m_1}  & h_{m_2}   & \ldots & h_{m_l}  \cr
   h_{m_1}\eta_1   & h_{m_2}\eta_2   & \ldots & h_{m_l}\eta_l  \cr
   h_{m_1}\eta_1^2 & h_{m_2}\eta_2^2 & \ldots & h_{m_l}\eta_l^2  \cr
   \vdots          & \vdots          & \ddots & \vdots          \cr
   h_{m_1}{\eta_1}^{l-1} & h_{m_2}{\eta_2}^{l-1} &
    \ldots & h_{m_l}{\eta_l}^{l-1} \cr
\end{pmatrix}
\begin{pmatrix}
e_{m_1} \cr e_{m_2} \cr e_{m_3} \cr \vdots \cr e_{m_l} \cr
\end{pmatrix}
=
\begin{pmatrix}
s_0 \cr s_1 \cr s_2 \cr \vdots \cr s_{l-1} \cr
\end{pmatrix}
\]
and then return $y-e$ (or \textit{Error} if one or more of the components
of $e$ is not in $K$).

\begin{rem}
Even with the improvements advanced in this note, 
in theory the PGZ and PGZm algorithms
cannot beat, for very large alternant codes, 
the Berlekamp-Massey-Sugiyama (BMS) algorithm (cf.  
\cite{Xambo-2003}, Section 4.3).
But they are comparable for the codes that
are feasible in practice. Indeed, the very construction of 
the alternant matrix is costly in time and space
and within the range of parameters that can usualy be afforded,
the efficiency of the PGZ or PGZm is comparable to that of BMS.
Let us also say that in some contexts, as for example in teaching, the PGZm
has the advantage that it is more straightfoward to explain and to implement,
the easiest case being RS codes over $K=\FF_p$, $p$ prime.
\end{rem}

\section{Examples}\label{sec:examples}
Here we are going to discuss the implementation of the algorithms 
using \cite{Sayols-Xambo-2017}, an how it works,
by considering some examples  
for each of the following classes: RS, GRS, BCH and
(classical) Goppa codes. 

In the code constructors described below, 
\textsf{h} and \textsf{a} stand for variables bound
to vectors of the same length $n$ with entries in a finite field;
\textsf{K} and \textsf{F}, to finite fields $K$ and $F$;
\textsf{r}, \textsf{k}, \textsf{d} and \textsf{l}, to integers $r$, $k$ 
$d$ and $l$ used as in the first two sections;
and \textsf{g} to a univariate polynomial with coefficients in a finite field.

\begin{itemize}
\item
\textsf{AC(h,a,r,K)}: This constructs the alternating code $A_K(\bs{h},\bs{\alpha},r)$. 
In the context of this note,
it is our main constructor, as the others (described below) are in fact defined as
special calls to \textsf{AC} (cf. Section \ref{sec:AC}). 
\item 
\textsf{RS(a,k)}: This yields the RS code $\RS(\bs{\alpha},k)$,
an $[n,k,n-k+1]$ code defined over the field to which
the elements of $\bs{\alpha}$ belong.

\item  
\textsf{GRS(h,a,k)}: As \textsf{RS}, but we have to supply $\bs{h}$ as a first
argument.
\item 
\textsf{PRS(F,k)}: The primitive RS code of the finite field $F$. It is defined as
\textsf{RS(a,k)}, but taking as
$\bs{\alpha}$ the list of non-zero elements of $F$.
\item 
\textsf{BCH(a,d,l)}: Supplies the code $\BCH(\alpha,d,l)$, where here 
\textsf{a} stands for an element $\alpha$ in a finite field.
\item 
\textsf{Goppa(g,a)}: The Goppa code $\Gamma(g,\bs{\alpha})$.
\end{itemize}

The code \textsf{C} obtained by any one of these constructors is a record-like structure
with fields that allow to get data from the code
or store new information about it. The labels of those fields end with an
underscore, but otherwise tend to mimic the mathematical symbols. For example,
\textsf{a = a\_(C)} and \textsf{H = H\_(C)} 
bind the variable \textsf{a} to the vector $\bs{\alpha}$ and the variable \textsf{H} to the alternant control matrix of $C$. 

\begin{rem}\label{rem:dimAC}
Except for RS and GRS codes, for which the parameters
can be deduced immediately from the data supplied to their constructors,
in general there is some work to be done to determine those not yet known. 
This work is rather straightforward when it comes to compute $k=\dim_K C$.
To that end, we need to construct a control matrix of $C$ defined over $K$.
This can be done in two steps: replace
each entry of $H$ by the column of is components in the
natural linear basis of $\bar{K}$ over $K$ 
(this yields an $K$-control matrix, but it usually has redundant rows) 
and then suppress all the rows that are linear combinations
of the previous ones. We have implemented these steps
by means of the functions \textsf{blow(H,K)} and \textsf{prune(M)}. 
The bottom line is that the dimension 
of $C$ is $n-r'$, where $r'$ is the number of rows of
\textsf{prune(blow(H,K))} or just \textsf{rank(blow(H,K))}. 
This is the method used 
to determine the dimension $k$ and the rate $R=k/n$
when we quote them. 
\end{rem}

A final comment before getting into the examples 
is that we can assume that the received vector is an error vector $e$
such that $\wgt(e)\le t$. The reason is the linearity of the code,
which implies that only the error vector is involved in the computations
of the error positions and values.

\newcommand{\lst}[1]{\par\vspace{5pt}\textsf{\small #1}\par\vspace{5pt}}
\renewcommand{\ni}{\noindent}

\paragraph{\bfseries RS}
Take $K=\FF_{13}$ and construct \textsf{C = PRS(K,8)}, the primitive RS of $K$ of dimension $k=8$. It has length $n=12$, so its rate is $2/3$, and the minimum distance
is $d=n-k+1=5$, so that it corrects at least two errors. First let us consider
the case of one error. Suppose the received vector is
\lst{e = [0, 0, 0, 0, 3, 0, 0, 0, 0, 0, 0, 0]}
\noindent Then the decoder call \textsf{PGZ(e,C)} (or \textsf{PGZm(e,C)}) yields
{\small 
\begin{verbatim}
   PGZ: Error positions [4], error values [3] :: Vector[Z13]
   [0, 0, 0, 0, 0, 0, 0, 0, 0, 0, 0, 0] :: Vector[Z13]
\end{verbatim}
}\ni
This means that \textsf{PGZ} finds that there is a single error
at the index 4 (5th element of the vector) and that its value is 3,
and then outputs (correctly) the decoded vector.
Let us go over the steps followed by \textsf{PGZ} in detail.
The control matrix is \textsf{H = H\_(C)}:
{\small
\begin{verbatim}
   [[1  2  4  8  3  6 12 11  9  5 10  7]
    [1  4  3 12  9 10  1  4  3 12  9 10]
    [1  8 12  5  1  8 12  5  1  8 12  5]
    [1  3  9  1  3  9  1  3  9  1  3  9]] :: Matrix[Z13]
\end{verbatim}
}\ni
Then the syndromy vector is given by \textsf{s = y*transpose(H)},
\lst{[9, 1, 3, 9]\textrm{,}}
\ni and the matrix \textsf{S = hankel\_matrix(s)} is
{\small 
\begin{verbatim}
   [[9 1 3]
    [1 3 9]] :: Matrix[Z13]
\end{verbatim}
}
\ni
This matrix has rank 1, as the second row is 3 times the first. 
This means that there is one error and
that the error-locating polynomial is $L(z)=z-1/9=z-3$. To find
the error position, we have to look at the position of 3 in the
vector $\bs{\alpha}$ used to construct \textsf{C}, which is \textsf{a\_(C)}:
\lst{[1, 2, 4, 8, 3, 6, 12, 11, 9, 5, 10, 7] :: Vector[Z13]}
\ni Thus the position is indeed the one in which the error occurred.
To find the error value, first we have to calculate the 
error evaluator 
\[
E(z) = \sigma(z)\wt{L}(z) \mod z^r = 
(9+ z+ 3z^2+ 9z^3)(-3z+1) \mod z^6,
\]
which turns out to be the constant $9$.
Forney's formula for the error value is
$-\alpha_4 E(1/\alpha_4)/h_4 \wt{L}^{'}(1/\alpha_4)=-3\cdot 9 /3\cdot(-3)=3$
(for in this case $\bs{h}=\bs{\alpha}$), which is the error value. 

Now we are going to repeat, with less detail, the case of 2 errors.
Suppose the received vector is
\lst{y = [0, 0, 0, 0, 3, 0, 0, 0, 0, 7, 0, 0]}
\noindent Then the decoder call \textsf{PGZ(y,C)} (or \textsf{PGZm(y,C)}) yields
{\small 
\begin{verbatim}
   PGZ: Error positions [4, 9], error values [3, 7] :: Vector[Z13]
   [0, 0, 0, 0, 0, 0, 0, 0, 0, 0, 0, 0] :: Vector[Z13]
\end{verbatim}
}
The syndromy vector is 
\lst{[5, 7, 7, 3]\textrm{,}}
\ni and the matrix \textsf{S = hankel\_matrix(s)} is
{\small 
\begin{verbatim}
   [[5 7 7]
    [7 7 3]] :: Matrix[Z13]
\end{verbatim}
}
\ni 
Since it has rank 2, there have been 2 errors.
In this case the Gauss-Jordan reduction produces 
$L(z)=z^2+5z+2=(z-3)(z-5)$. Since the roots 3 and 5 occupy
the positions 5 and 10 in $\bs{\alpha}$, we see that
the indices of the computed error positions are correct.
For the error values we have to apply Forney's formula
to 3 and 5 ($\alpha_4$ and $\alpha_9$). 
We have $\wt{L}(z)=2z^2+5z+1$, $\wt{L}'(z)=4z+5$,
$\sigma(z)=5+7z+7z^2+3z^3$, and 
$E(z)=\wt{L}(z)\sigma(z) \mod z^4=6z+5$. Then
the error corresponding to, for example, the second root is
\[
-5 \cdot E(1/5)/5 \cdot \wt{L}'(1/5)=-E(8)/\wt{L}'(8)
=-(48+5)/(32+5)=-1/-2=7.
\]

For another example, if we take \textsf{F = Zn(31)} and \textsf{k = 20}, then
\textsf{C = PRS(F,k)} is a $[30,20,11]$ code. This corrects up to
5 errors and its rate is $2/3$. 
This capability is illustrated in the following listing:

{\footnotesize
\begin{verbatim}
   e = rd_error_vector(F,n,5)  # this creates a ramdom 5-error pattern
   >>[0,0,0,0,0,0,0,0,0,14,0,0,0,28,26,0,0,0,0,23,0,0,16,0,0,0,0,0,0,0]
         :: Vector[Z31]
   PGZ(e,C)
   >>PGZ: Error positions [9,13,14,19,22], 
          error values [14,28,26,23,16] :: Vector[Z31]
\end{verbatim}
}

\paragraph{\bfseries BCH}
Take $K=\FF_2$ and
$F=\FF_{32}$, generated by $\alpha$ such that
$\alpha^5=\alpha^2+1$. Let $C=\BCH(\alpha,7)$.
This is a binary code of length $n=31$ (the
order of $\alpha$) that corrects up to 3 errors.
In our system it can be constructed as follows:
{\small 
\begin{verbatim}
   K = Zn(2)
   [F,a] = extension(K,[1,0,0,1,0,1],'a','F')
   C = BCH(a,7)
\end{verbatim}}
\ni Its dimension is 16 (so its rate is $16/31>1/2$), as shown by the following
command:
{\small 
\begin{verbatim}
   n - rank(blow(H_(C),K))
   >> 16
\end{verbatim}
}
Now consider, for example, the weight 3 error pattern \textsf{e}:
\lst{%
[0, 0, 0, 0, 0, 1, 0, 0, 0, 0, 0, 0, 0, 0, 0, 0, 0, 0, 0, 1, 0, 0, 0, 0, 0, 0, 0, 0, 1, 0, 0]
}\ni
Then the call \textsf{PGZ(e,C)} outputs 
\lst{PGZ: Error positions [5, 19, 28], error values  [1, 1, 1] :: Vector[K]}
\ni (Note that the only possible error value over $\FF_2$ is 1, and that 
therefore in this case the decoder only needs to care about 
error location.)
The matrix $S$ computed in this case is (instead of $\alpha^j$ we write $j$):
{\small 
\begin{verbatim}
    [[22, 13, 14, 26]
     [13, 14, 26, 19]
     [14, 26, 19, 28]]
\end{verbatim}
}\ni
which gives $l=3$.

The code $C$ also corrects up to 3 errors
when considered as an $F$-code (which is a GRS code over $F$). 
For example, if
\lst{%
e=[0,0,0,0,0,0,0,0, $\alpha^5$, 1, 0,0,0,0,0,0,0,0,0,0,0,0,0,0,0,0, 
$\alpha^{19}$, 0,0,0,0]
}\ni
then the output contains
{\small 
\begin{verbatim}
   PGZ: Error positions [8, 9, 26], 
        error values [a**5, 1, a**19] :: Vector[F]
\end{verbatim}}
\ni together with the correct decoded vector. With the
same conventions as before, the matrix $S$ gives $l=3$:
{\small 
\begin{verbatim}
    [[16,  0, 30, 14]
     [ 0, 30, 14, 25]
     [30, 14, 25, 28]]
\end{verbatim}
}

Here is a more involved example. Take $K=\FF_3$ and
$F=\FF_{343}$, generated by $\alpha$ such that
$\alpha^5=\alpha+2$. The element $\alpha$ is 
primitive and $\beta=\alpha^2$ has order $n=(343-1)/2=121$. 
It follows that the code $C=\BCH(\beta,11)$
has length $n$ and that it corrects at least 5 errors.
In the \textsf{PyECC} system it can be constructed as folows:
{\small 
\begin{verbatim}
   K = Zn(3)
   f=get_irreducible_polynomial(K,5,'X')  # X**5-X+1
   [F,a] = extension(K,f,'a','F')
   C = BCH(a**2,11)
\end{verbatim}}\ni
In addition, the command \textsf{n - rank(blow(H\_(C),K))}
yields that its dimension is $121-35 = 86$, so that
its rate is $86/121>7/10$.

Consider the error pattern of weight 5 (where \textsf{0$_{k}$} denotes
\textsf{0} repeated \textsf{k} times)
\lst{%
e = [0$_2$, 1, 0$_7$, 1, 0$_{22}$, 2, 0$_6$, 2, 0$_{72}$, 1, 0$_{7}$] :: Vector[Z3]
}
Then we have:
{\small 
\begin{verbatim}
   PGZ(e,C)
   >>PGZ: Error positions [2, 10, 33, 40, 113],
          error values [1, 1, 2, 2, 1] :: Vector[F5]
\end{verbatim}
}

\paragraph{\bfseries Classical Goppa}
Consider the field $\FF_{25}$ generated over $\FF_5$ by $x$
such that $x^2=2$. Let $g=T^6+T^3+T+1$ and make a list $\bs{\alpha}$
of the elements $t\in \FF_{25}$ such that $g(t)\ne 0$. Then it turns out that
$\bs{\alpha}$ has length $n=19$
($g$ has four simple roots and one double root in $\FF_{25}$)
and that $C=\Gamma(g,\bs{\alpha})$ 
corrects up to 3 errors. 

{\small
\begin{verbatim}
   F5 = Zn(5)
   # Creation of F25, with generator x
   [F25,x] = extension(F5,[1,0,-2],'x','F25') 
   # Creation of the polynomial ring F25[T]
   [A,T] = polynomial_ring(F25,'T')

   g = T**6 + T**3 + T +1
   a = Set(F25)[1:]		# The non-zero elements of F25 
   a = [t for t in a if evaluate(g,t)!=0]
   C = Goppa(g,a)

   # generate a random error pattern of weight 3
   e = rd_error_vector(Z5,n,3)
   >> e = [0,1,0,0,0,3,0,4,0,0,0,0,0,0,0,0,0,0,0] :: Vector[Z5]

   # Use the PGZ decoder for C
   PGZ(e,C)
   >>PGZ: Error positions [1,5,7], error values [1,3,4] :: Vector[K]
     [0,0,0,0,0,0,0,0,0,0,0,0,0,0,0,0,0,0,0] :: Vector[Z5]
\end{verbatim}
}\ni
The dimension of $C$ (given by 
\textsf{n - rank(blow(H\_(C),F5)}) is $19-12=7$, and its rate
$7/19>1/3$.

\medskip
Here is another illustration, with $n=76$:
\lst{%
e = rd\_error\_vector(F3,n,5) $>\!\!>$ [0$_{10}$, 2, 0$_{35}$, 2, 0$_{9}$, 1, 0$_6$, 1, 0$_3$, 2, 0$_{10}$]
}
\vspace{-2ex}
{\small 
\begin{verbatim}
   F3 = Zn(3)
   f = get_irreducible_polynomial(F3,4,'X') # X**4 + X + 2
   [F81,x] = extension(F3,f,'x','F81')      # x is primitive     
   g = X**2 * (X-1)**4 * (X-2)**4           
   a = Set(F81)[3:]   # g does not vanish on any of these values 
   n = len(a)         # 81 - 3 = 76                  
   C = Goppa(g,a)     # code of length 76
   k = n - rank(blow(H_(C),F3))   # k = 76-32 = 44 
   PGZm(e,C)
   >> PGZm: Error positions [10, 46, 56, 63, 67], 
           error values [2, 2, 1, 1, 2] :: Vector[F3]
\end{verbatim}
}

\appendix

\section{The function PGZm}
For the sake of brevity, here we list and comment the \textsf{PGZm(y,C)} function.
Its code, together with the code of 
\textsf{PGZ(y,C)}, can be accessed by browsing the text file
\href{https://mat-web.upc.edu/people/sebastia.xambo/Listings-FSX/}{Listings-FSX.pyecc}.
The parameter \textsf{y} is supposed to be the received vector in a transmission
using the alternant code \textsf{C}. We have seen that the value of expressions
\textsf{a\_(C)} and \textsf{H\_(C)} is the vector $\bs{\alpha}$ and the control
matrix $H$. Similarly, the values of the expressions
\textsf{K\_(C)}, \textsf{h\_(C)}, \textsf{r\_(C)}
are the field over which \textsf{C} is defined, the vector $\bs{h}$ and
the number of rows of $H$, respectively.

{\small 
\begin{verbatim}
def PGZm(y,C):
    if isinstance(y,list): y = vector(K_(C),y)
    if not isinstance(y,Vector_Element):
        return "PGZm: Argument is not a vector"
    h = h_(C)
    if len(y) != len(h):
        return "PGZm: Vector argument has wrong length"
    r = r_(C); alpha = a_(C); H = H_(C); K = K_(C)
    s = y*H.transpose()   
    if is_zero(s): 
        print("PGZm: Input is a code vector")
        return y
    S = hankel_matrix(s)
    c0 = S[:,0]   # keep the first column of S
    a = -GJ(S); l = len(a)
    a = reverse(a.to_list())   
    K1 = K_(H)
    [_,z] = polynomial_ring(K1,'z','K1[z]')
    L = hohner([1]+a,z)
    R = [s for s in alpha if evaluate(L,s)==0]
    if len(R) < l:
        return "PGZm: Defective error location"
    M = [alpha.to_list().index(r) for r in R]
    h1 = [h[m] for m in M]
    V = alternant_matrix(h1,R,l)
    v = c0[:l]
    V1 = splice(V,v)
    w = transpose(GJ(V1))
    for t in w:
        t = pull(t,K)
        if not belongs(t,K): 
            return "PGZ: error value not in base field"
    show("PGZm: Error positions {}, error values {}".format(M, w))
    for j in range(len(M)):
        y[M[j]]-=w[j]
    return pull(y,K)
\end{verbatim}
}

\section{The PyECC system}

Initially (October 2015) the idea that launched 
\cite{Sayols-Xambo-2017}
was to match the functionality of the CC system
developed to deal with the computational tasks related to the book \cite{Xambo-2003},
but it became soon clear that we could go beyond that system in several
directions. The aim of the undertaking is to
produce a Python package (PyECC) enabling the construction, coding and decoding of
error-correcting codes and make it freely available for teachers and
researchers. The current state of the project 
is documented at \href{https://mat-web.upc.edu/people/sebastia.xambo/PyECC.html}{PyECC}

\end{document}